\documentclass[conference, a4paper]{IEEEtran}
\IEEEoverridecommandlockouts

\usepackage{multirow}
\usepackage{lipsum}
\usepackage{amsfonts}
\usepackage{array}
\usepackage{amsmath,cases,amssymb}
\usepackage{amsmath}
\usepackage{amsthm}
\usepackage{graphicx}
\usepackage{cite}
\usepackage{subfigure}
\usepackage{epsfig}
\usepackage{url}
\usepackage{algorithm}
\usepackage{algorithmic}
\usepackage{epstopdf}
\usepackage{balance}
\usepackage{color}
\usepackage{makecell}
\usepackage{sansmath}

\DeclareGraphicsExtensions{.eps,.pdf}

\newcommand{\maxi}{{\mathrm{maximize}}}

\newcommand{\bw}{\mathbf{w}}

\newcommand{\bH}{\mathbf{H}}

\newcommand{\bI}{\mathbf{I}}
\newcommand{\bh}{\mathbf{h}}

\newcommand{\bW}{\mathbf{W}}

\newcommand{\Qcal}{\mathcal{Q}}
\newcommand{\Kcal}{\mathcal{K}}

\newcommand{\UEk}{\mathtt{UE}_k}

\makeatletter
\newcommand{\doublewidetilde}[1]{{%
		\mathpalette\double@widetilde{#1}}}
\newcommand{\double@widetilde}[2]{%
	\sbox\z@{$\m@th#1\widetilde{#2}$}%
	\ht\z@=.5\ht\z@
	\widetilde{\box\z@}}
\makeatother

\newtheorem{theorem}{Theorem}
\newtheorem{lemma}{Lemma}
\newtheorem{corollary}{Corollary}

\newtheorem{remark}{Remark}
\newtheorem{assumption}{Assumption}
\begin{document}
	
	\title{\huge Learning to Optimize: Balancing Two Conflict Metrics in MB-HTS Networks \vspace*{0cm}}
	
	\author{Van-Phuc~Bui$^*$, Trinh~Van~Chien$^\xi$,  Eva~Lagunas$^*$,  Jo\"el~Grotz$^\dag$, Symeon~Chatzinotas$^*$,  and Bj\"orn~Ottersten$^*$ 
		\\
		$^*$Interdisciplinary Centre for Security, Reliability and Trust (SnT), University of Luxembourg, Luxembourg\\
		$^\xi$School of Information and Communication Technology, Hanoi University of Science and Technology, Vietnam \\
		$^\dag$SES Engineering, Luxembourg, Luxembourg
		\thanks{This work has been supported by the Luxembourg National Research Fund (FNR) under the project INSTRUCT (IPBG19/14016225/INSTRUCT). 
		}
		
	}
	\maketitle
	\begin{abstract}
		\textcolor{black}{For multi-beam high throughput (MB-HTS) geostationary (GEO) satellite networks, the congestion appears when user's demands cannot be fully satisfied. This paper boosts the system performance by formulating and solving the power allocation strategies under the congestion control to admit users. A new multi-objective optimization is formulated to balance the sum data throughput and the satisfied user set. After that, we come up with two different solutions, which efficiently tackle the multi-objective maximization problem: The model-based solution utilizes the weighted sum method to enhance the number of demand-satisfied users, whilst the supervised learning solution offers a low-computational complexity design by inheriting optimization structures as continuous mappings.  Simulation results verify that our solutions effectively copes with the congestion and outperforms the data throughput demand than the other previous works.}
	\end{abstract}

	\section{Introduction}\label{sec:intro}
	\textcolor{black}{An effective solution giving high-speed broadband data throughput to users has been recorded in  multi-beam high throughput satellite (MB-HTS) networks,  applicable for a large area that is inaccessible or under-served locations by the terrestrial systems \cite{Kodheli:21:tut}. Related works in the literature related to precoded MB-HTS have been observed, while many of them including the data throughput constraints regarding the minimum throughput for the users' weakest channel conditions \cite{ZhangAsilomar2020}. To guarantee the individual data throughput demand, the work in \cite{Zheng:twc:12} studied an optimal design in a multi-beam satellite system via using an alternating optimization method. In spite of data throughput improved along the iterative procedure, the proposed solution in \cite{Zheng:twc:12} has the high computational complexity and low channel rate because the max-min fairness optimization can not ensure an acceptable data throughput to all the users in a large-scale system. Linear precoding as  regularized zero-forcing (RZF) or zero-forcing (ZF) has manifested the high sum throughput with a low cost in MB-HTS networks \cite{trinh2021user}. Nonetheless, the previous works are all based on non-empty feasible regions to ensure the existence of a solution. For a complicated network serving many users with divergent data throughput, only one user having very weak  channel condition due to obstacles or the high requested QoS might be sufficient to make the optimization problem infeasible under limited radio resources.}

	\textcolor{black}{The potentialities of machine learning has approved in designing data-driven solutions for challenging problems in data analysis and wireless communications thanks to the exploitation of neural networks \cite{Lei:Access:20, Chien19:TCOM}. In particular,  the difficulties from the NP-hardness of different beam hopping  were effectively tackled with a machine learning-based solution giving a high accuracy in \cite{Lei:Access:20}. For the resource allocation in satellite communications, e.g., the power controls, the work in \cite{Chen:cof:20} minimized the total transmit power consumption with respect to the data throughput requirements for multi-beam satellite network without using a precoding technique. Note that the above related works only optimized the single objectives without the congestion issues, where one or several users may be not served with their demands and making conflicts to the remaining users.}
	
	\textcolor{black}{In this paper, we investigate a new category of multi-objective problems operating formulated for the MB-HTS GEO networks. The objective is to maximize the number of demand-satisfied users and the sum data throughput. By deploying the weighted sum method, a general model-based solution converts the multi-objective problem in the original form to a corresponding single-objective problem by balancing between the two utility metrics. Via exploiting the low-computational complexity of the linear beamforming methods, we can, therefore, design  low-cost algorithms for the considered MB-HTS systems by altering the general model-based solution. 
	To bring the framework closer to real time resource allocation, we further develop a data-driven solution where a supervised learning scheme is utilized to train a neural network and then predicting the transmit power assigned to the users and the demand-satisfied set as well. Our proposed algorithms are testified by a satellite beam-pattern from the European Space Agency \cite{ESA} that emulated a Defocused Phased Array Fed Reflector (PAFR) antenna array.
	Our solutions are compared with the previous approaches  \cite{lu2019robust,trinh2021user, Krivochiza:21:access} in the satellite communication literature with subject to both the data throughput and users' demand satisfaction.}
	
	\textcolor{black}{\textit{Notation}: Upper and lower bold letters to denote matrix and vectors. $\| \cdot \|$ denotes the Euclidean norm, while $| \mathcal{X} |$ is the cardinality of  set $\mathcal{X}$. The Hermitian and regular transposes are represented by $(\cdot)^H$  and $(\cdot)^T$. The circularly symmetric Gaussian distribution is  $\mathcal{CN}(\cdot, \cdot)$, while the expectation operator is $\mathbb{E}\{ \cdot\}$. The identity matrix of size $K \times K$ is $\mathbf{I}_K$. The element-wise inequality is $\succeq$ and $\mathbf{1}_K$ is a unit vector of length $K$. $\mathrm{tr}(\cdot)$ is the trace of a matrix.  The sets $\mathbb{R}$, $\mathbb{R}_{+}$, $\mathbb{R}_{++} = \mathbb{R}_{+} \cup \varnothing$, and $\mathbb{C}$ denote the  real, non-negative real,  extended non-negative real, and complex field respectively.}
	
	\section{MB-HTS System Model and Optimization Problem Formulation}\label{sec:sys_model}
	\subsection{GEO Satellite Model \& Channel Capacity}
	A broadband MB-HTS network is considered in the forward link where multiple users simultaneously joint the network with the same time and frequency resource. In the coverage area, we assume that there are $N$ overlapping beams providing services to a maximum of $N$ users. Since these users can be scheduled and served  by the satellite in each time slot, the actual users per scheduling instance is $K$ with $K \leq N$. Each user~$k$ with  $k \in \Kcal \triangleq \{1, 2, \dots, K\}$  and $|\Kcal| = K$ is simply denoted as $\UEk$. We define the vector $\bh_k\in\mathbb{C}^{N}$ to represent the propagation channel between the satellite and $\UEk$. Let us denote ${\bH} = [{\bh}_1,{\bh}_2,\dots,{\bh}_K] \in \mathbb{C}^{N\times K}$ the matrix that gathers the channel state information (CSI) as
		${\bH} = \bar{\bH}\mathbf{\Phi}$,
		where $\bar{\bH}\in\mathbb{R}_+^{N\times K}$ represents  practical features comprising, for example, the antenna radiation pattern of satellite antennas, additive noise, received antenna gain at users, and path loss from a long propagation distance. Specifically, the $(n,k)$-th element of $\bar{\bH}$ is given as
				\begin{equation}
		[\bar{\bH} ]_{nk} =\frac{\lambda\sqrt{G_R G_{nk}}}{4\pi d_k\sqrt{K_BTB}},
				\end{equation}
		where $\lambda$ denotes the wavelength of the carrier wave; $d_k$ represents the distance between $\UEk$ and the satellite; $G_R$ and $G_{nk}$ denotes the gains at the receiver antenna gain and from the $n$-th satellite feed towards $\UEk$, $\forall n = 1, \ldots N$; $K_B$ denotes the Boltzmann constant; $T$ represents  the noise temperature at the receiver. The diagonal matrix $\mathbf{\Phi}\in\mathbb{C}^{K\times K}$ stands for the signal phase rotations originated from the antenna architecture with 
	the $(k,l)$-th element defined as $[\mathbf{\Phi}]_{kk}= e^{j\phi_k}$,
		where $\phi_k$ denotes a residual random phase component from the satellite payload. Let us denote $s_k$  the data symbol that the network transmits to $\UEk$ with $\mathbb{E} \{ |s_k|^2 \} = 1$ and its transmit power $p_k \in \mathbb{R}_{+}$. Defining $\mathbf{w}_k \in \mathbb{C}^N$ the normalized precoding vector for $\UEk$ with $\| \bw_k\| = 1$, and then the transmitted signal $\mathbf{x} \in \mathbb{C}^N$ is formulated as
		\begin{equation} \label{eq:transsig}
	\mathbf{x} = \sum\limits_{k\in\Kcal} \sqrt{p_k} \mathbf{w}_k s_k.
		\end{equation}
	The received signal  $y_k \in \mathbb{C}$ at $\UEk$ is  given by
	\begin{equation} \label{eq:y}
		y_k = \mathbf{h}_k^H \mathbf{x} =  \sqrt{p_k} \bh_k^H\bw_k s_k + \sum\limits_{\ell\in\Kcal\backslash \{k\}} \sqrt{p_\ell} \bh_k^H\bw_\ell s_\ell + n_k,
	\end{equation}
	\textcolor{black}{where $n_k$ is additive noise with $n_k \sim \mathcal{CN}(0, \sigma^2)$. We assume that the perfect CSI available at the gateway, so the channel capacity of $\UEk$, measured in [Mbps], is given as
	\begin{equation}\label{eq:rate}
		\fontsize{10}{10}{R_k ( \{ p_{k'} \} ) = B \log_2 \Big(1+ \frac{p_k|\bh_k^H\bw_k|^2}{\sum_{\ell\in\Kcal\backslash \{k\}}p_\ell |\bh_k^H\bw_\ell|^2 +\sigma^2} \Big),} 
	\end{equation}
	where $\{ p_{k'} \} = \{ p_1, \ldots, p_K \}$ contains all the transmit powers and $B$~[MHz] denotes the system bandwidth utilized for the user link. The channel rate expression~\eqref{eq:rate} is applicable for an arbitrary channel model and precoding method.} 
	\subsection{Single Objective With Data Throughput Constraints}
	\textcolor{black}{Regarding the MB-HTS networks, conventional power control problems concentrating on maximizing a utility function, whilst maintaining the data throughput requirements of the users with a finite transmit power level. Considering the sum data throughput as an objective function, a popular optimization problem \cite{Bjon13:tcit} is formulated as
	\begin{subequations} \label{RelatedWork}
		\begin{alignat}{2}
			&\underset{\{ p_{k'} \in \mathbb{R}_{+} \}}\maxi &\quad & f_0(\{ p_{k'} \}) \triangleq  \sum\nolimits_{k\in\Kcal}R_k( \{ p_{k'} \} ) \label{RelatedWorka}\\
			&\mbox{subject to} && R_k(\{ p_{k'} \})  \geq \xi_k, \mbox{ } \forall k\in\Kcal, \label{RelatedWorkb}\\
			&&& \sum\nolimits_{k\in\Kcal} p_{k} \leq P_{\max}, \label{RelatedWorkc}
		\end{alignat}
	\end{subequations}
	where $\xi_k$ [Mbps] represents the data throughput required by $\UEk$ and $P_{\max}$ is the maximum transmit power that the satellite can provide. The network may not offer the data throughput requirements to every user, which leads to the congestion and at least one or several users served less data throughput than what they requested. The congestion makes problem~\eqref{RelatedWork} infeasible  because of an empty feasible domain. 
	\begin{figure}[t]
		\begin{minipage}{0.15\textwidth}
			\includegraphics[trim=0.cm 0.0cm 0.cm 0.8cm, clip=true, width=1.1in]{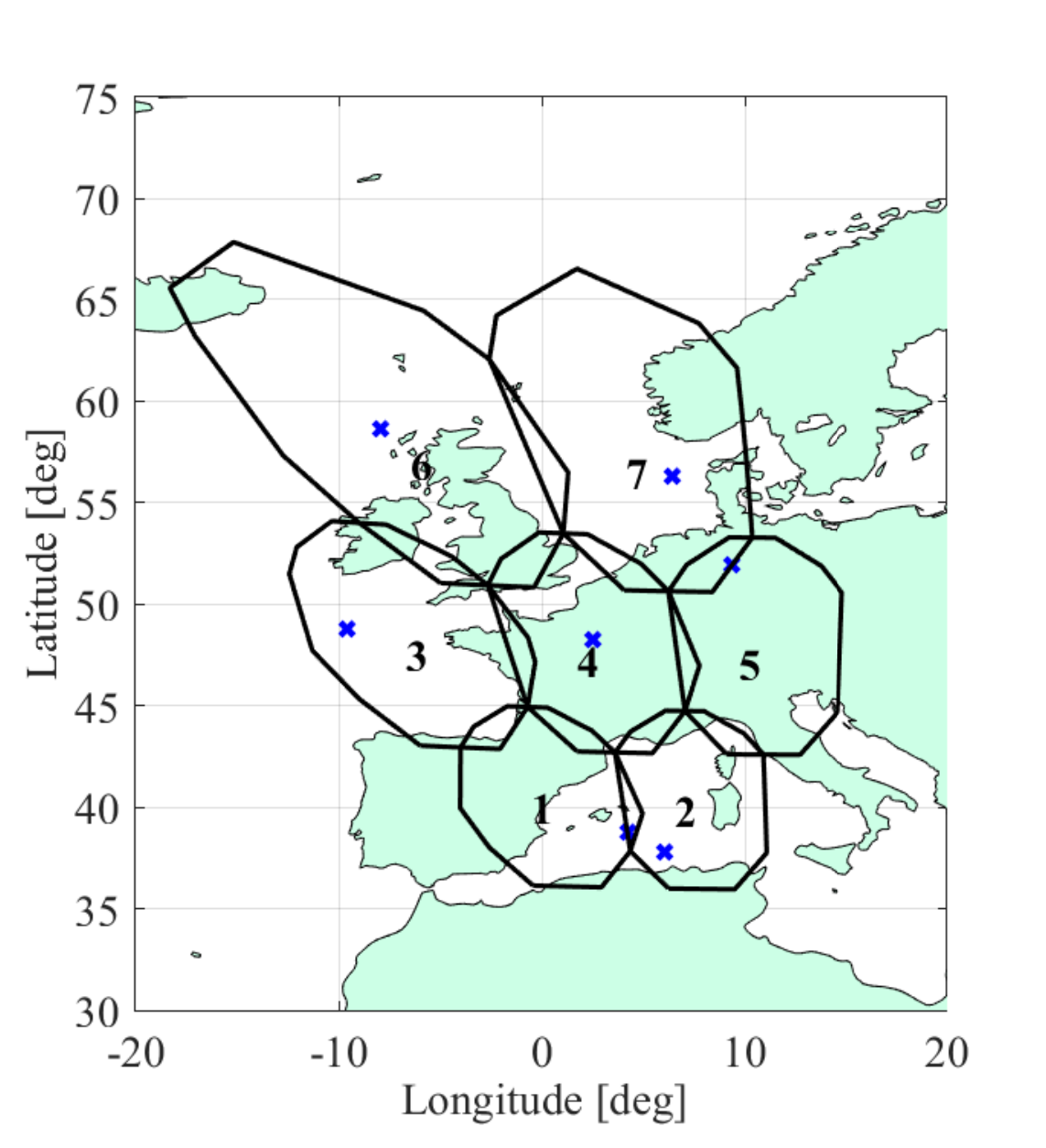} \\ 
			\centering \fontsize{8}{8}{$(a)$}
		\end{minipage}
		\begin{minipage}{0.15\textwidth}
			\includegraphics[trim=0.0cm 0.2cm 0.0cm 0.8cm, clip=true, width=1.1in]{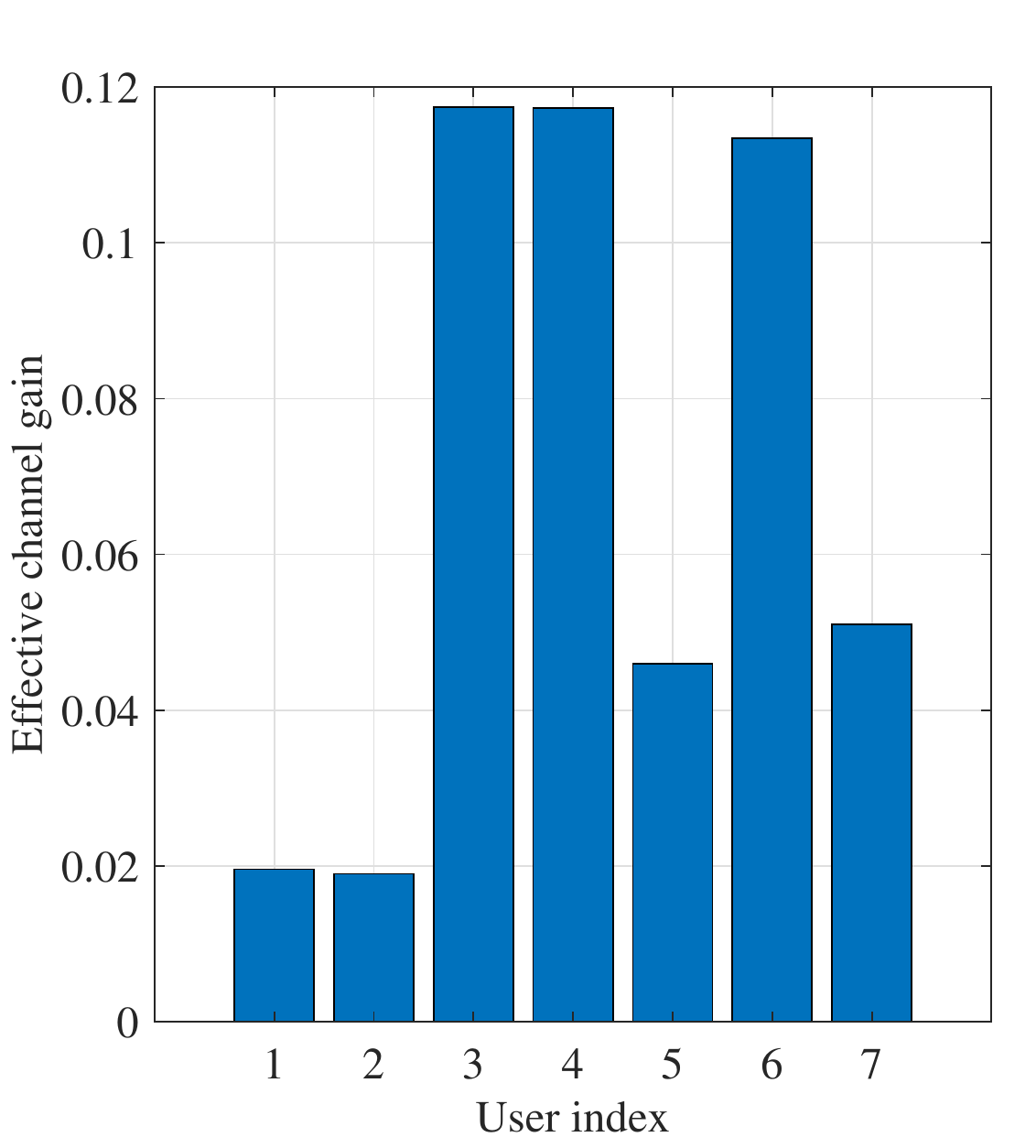}\\
			\centering \fontsize{8}{8}{$(b)$}
		\end{minipage}
		\begin{minipage}{0.15\textwidth}
			\includegraphics[trim=0.0cm 0.2cm 0.0cm 0.8cm, clip=true, width=1.1in]{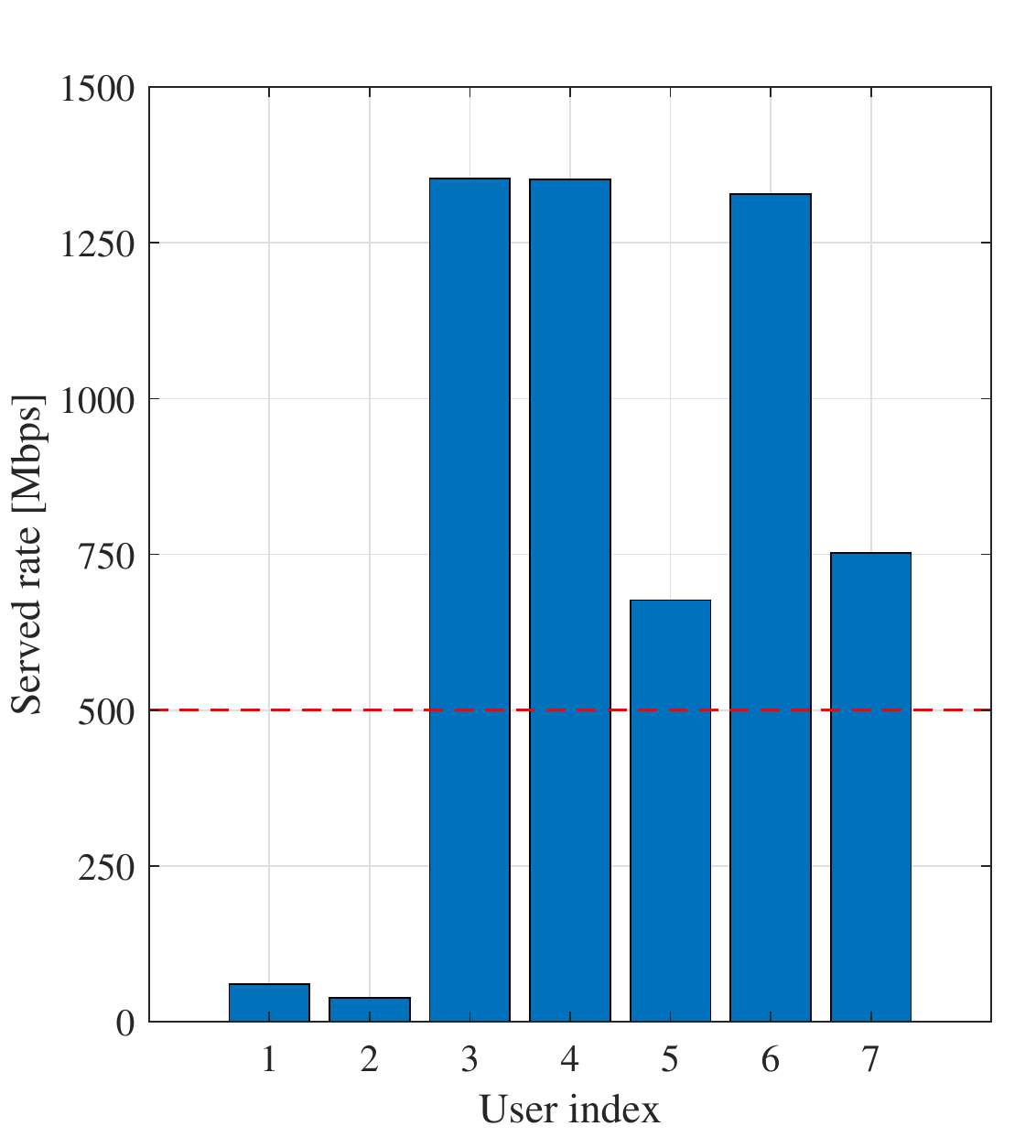} \\
			\centering \fontsize{8}{8}{$(c)$}
		\end{minipage}
		\caption{\textcolor{black}{A scheduling example where each beam serves one user: $(a)$ plots the user locations; $(b)$ plots the effective channel gains formulated as $|\mathbf{h}_k^H \mathbf{w}_k|^2, \forall k \in \Kcal$; and $(c)$ plot the served data throughput [Mbps] by the ZF precoding method. The detailed parameters are provided in Section~\ref{Sec:Results}.}}
		\label{Fig:ChannelvsRate}
	\end{figure}
	For tractability, one may formulate an optimization neglecting the data throughput constraints as}
	\begin{subequations} \label{RelatedWorkSumRate}
	\begin{alignat}{2}
				&\underset{\{ p_{k'} \in \mathbb{R}_{+} \}}\maxi &\quad & \sum\nolimits_{k\in\Kcal}R_k( \{ p_{k'} \} )\\
				&\mbox{subject to} && \sum\nolimits_{k\in\Kcal} p_{k} \leq P_{\max}, 
		\end{alignat}
	\end{subequations}
	which was addressed in \cite{aravanis2015power} and the mentioned references therein.
	Fig.~\ref{Fig:ChannelvsRate}(a) illustrates an instance of $N=7$ overlapping beams with $K=7$ users.  Fig.~\ref{Fig:ChannelvsRate}(c) displays the achievable data throughput for each of the users by solving problem~\eqref{RelatedWorkSumRate} utilizing the effective channel gains defined in Fig.~\ref{Fig:ChannelvsRate}(b). Notably, two users having very weak effective channel conditions along with a finite power resource will make it difficult for the satellite to guarantee all the users to be jointly served with the individual data throughput requirement, e.g., $500$ [Mbps]. Nonetheless, the remaining users still get at least their requested data throughput. 
	\subsection{Proposed Multi-Objective Optimization} \label{sec:MultiObject}
	\textcolor{black}{To handle the congestion issues, we divide the $K$ users into two sets: $\Qcal$ with  $\Qcal \subseteq \Kcal$ are the demand-satisfied user set containing users served by the
	data throughput at least their requirements. The remaining users are grouped in the demand-unsatisfied user set $\Kcal \setminus \Qcal$. Our target is to maximize the cardinality of  $\Qcal$ and further to find the maximum sum data throughput metric $\sum_{k\in\Kcal}R_k( \{ p_{k'} \} )$ as 
	\begin{subequations} \label{probGlobal}
		\begin{alignat}{4}
			&\underset{ \{ p_{k'} \in \mathbb{R}_{+} \}, \Qcal}{\mathrm{maximize}}&\ & \mathbf{g}\left(\{ p_{k'}\}, \Qcal \right) = \Big[ \sum\nolimits_{k\in\Kcal}R_k( \{ p_{k'} \} ),  |\Qcal| \Big]^T \label{probGlobala}\\
			&\mbox{subject to} && R_k( \{ p_{k'} \}) \geq \xi_k,  \forall k\in\Qcal, \label{probGlobalb}\\
			&&& \sum\nolimits_{k \in\Kcal} p_k \leq P_{\max},\label{probGlobalc}\\
			&&& \Qcal  \subseteq \Kcal. \label{probGlobald}
		\end{alignat}
	\end{subequations}
	Observing that \eqref{probGlobal} is always feasible because $\Qcal$ can span from an empty set to the demand-satisfied set $\mathcal{K}$. The constraints \eqref{probGlobalb} only ensure the individual data throughput requirements of $\Qcal$. Unlike a single objective function defined in \eqref{RelatedWork}, the decision space of problem~\eqref{probGlobal} is formulated as
	$\mathcal{D} =  \big\{ \{ p_{k'}\}, \Qcal  \big|  R_k( \{ p_{k'} \}) \geq \xi_k, \forall k\in\Qcal, 
	P_{\max} \geq \sum\nolimits_{k \in\Kcal} p_k , \Qcal \subseteq \Kcal \big\}$,
	which is non-convex. The data of \eqref{probGlobal} includes the decision space $\mathcal{D}$, the two objective functions contained in $\mathbf{g}\left(\{ p_{k'}\}, \Qcal \right)$, and the objective space $\mathbb{R}_{++}^2$. Mathematically, we map $\mathbf{g}\left(\{ p_{k'}\}, \Qcal \right)$ from the objective space to an ordered space, i.e., $(\mathbb{R}_{++}^2, \geq, \subseteq)$, where the feasibility is iteratively verified  by the order relations $\geq$ and $\subseteq$. The mapping is denoted as the $\theta$ model, describing a relation between the objective and order spaces, in which the maximization in \eqref{probGlobal} should be defined. Alternatively, \eqref{probGlobal} is explicitly characterized by the data $(\mathcal{D}, \mathbf{g}(\{ p_{k'} \}, \Qcal), R_{++}^2)$, the model map $\theta$, and the order space $\mathbb{R}_{++}^2$. We examine  an $\pmb{\epsilon}$-\textit{Pareto optimal} solution $\{ \{ p_{k'}^\ast \}, \Qcal^{\ast} \} \in \mathcal{D}$  to \eqref{probGlobal}, if there exists no $\{ \{ p_{k'} \}, \Qcal \} \in \mathcal{D}$ such that
	$\mathbf{g}\left(\{ p_{k'} \}, \Qcal \right) + \pmb{\epsilon} \succeq  \mathbf{g}\left(\{ p_{k'}^\ast \}, \Qcal^\ast \right)$, 
	where $\pmb{\epsilon} = [\epsilon_1, \epsilon_2]^T$ with $\epsilon_1, \epsilon_2 \in \mathbb{R}_+$ related to the accuracy of two considered objective functions.
	It indicates no better solutions $\{ \{ p_{k'} \}, \Qcal \} \in \mathcal{D}$ satisfied the conditions:
	\begin{equation}
		\fontsize{10}{10}{f_0\left(\{ p_{k'} \}, \Qcal \right) + \epsilon_1 \geq f_0 \left(\{ p_{k'}^\ast \}, \Qcal^\ast \right)  \mbox{;   } |\Qcal| + \epsilon_2 \geq |\Qcal^\ast|,}
	\end{equation}
	which exposes a trade-off between the two considered objective functions at the optimum. If $\epsilon_1 = \epsilon_2 = 0$, the above definition becomes an $\pmb{\epsilon}$-Pareto optimal solution
	that should be only enhanced by improving one objective function and reducing the other. Hence, an $\pmb{\epsilon}$-\textit{properly Pareto optimal solution} is defined as an $\pmb{\epsilon}$-Pareto optimal solution with a bound trade-off between the two considered objectives given in \eqref{probGlobala}. An $\pmb{\epsilon}$-\textit{Pareto dominant vector} is formulated as the objective functions in $\mathbf{g}(\{ p_{k'} \}, \Qcal)$ at the  $\pmb{\epsilon}$-properly Pareto optimal solution. We stresss that the $\pmb{\epsilon}$-\textit{Pareto frontier} gathers all the properly $\pmb{\epsilon}$-Pareto optimal vectors.}

	\section{Model-based and Data-driven Solutions} \label{Sec:ModelDL}
	\textcolor{black}{This section proposes the solutions to attain an $\pmb{\epsilon}$-properly Pareto optimum of problem~\eqref{probGlobal}.}
	\subsection{Model-based Solution}
	\textcolor{black}{We handle \eqref{probGlobal} by utilizing the weighted sum method \cite{ehrgott2005multicriteria}. Let us  define the weights $\mu_1 \geq 0$ and $\mu_2 \geq 0$ satisfying  $\mu_1 + \mu_2 = 1$ that shows the priority of the two objective functions in $\mathbf{g}(\{ p_{k'} \}, \Qcal)$. If $\{ \{ p_{k'}^{\ast} \}, \Qcal^\ast \}$ is an optimal solution to the following single-objective optimization problem:
	\begin{subequations} \label{probGlobalv1}
		\begin{alignat}{2}
			&\underset{\{ p_{k'} \in \mathbb{R}_{+} \}, \Qcal}{\mathrm{maximize}} \quad \mu_1 \sum\nolimits_{k\in\Kcal}R_k( \{ p_{k'} \} ) +  \mu_2 |\Qcal| \label{probGlobalv1a}\\
			&\mbox{subject to} \quad \ \eqref{probGlobalb}, \eqref{probGlobalc}, \eqref{probGlobald},
		\end{alignat}
	\end{subequations}
	for a given $\pmb{\epsilon}$-accuracy, and then $\{ \{ p_{k'}^{\ast} \}, \Qcal^\ast \}$  is an $\pmb{\epsilon}$-properly  Pareto  optimum to  problem~\eqref{probGlobal}. An $\pmb{\epsilon}$-Pareto frontier to \eqref{probGlobal} can be further attained by flexibly selecting  $\mu_1$ and $\mu_2$ 	in \eqref{probGlobalv1}.}
	\begin{assumption}\label{AssumpPrior}
	\textcolor{black}{The weights $\mu_1$ and $\mu_2$ are selected to attend the largest cardinality of the set $\Qcal$, i.e., offering the data throughput requirements for a maximum number of users, and then paying attention to maximize the sum-rate. The finite power budget and the weak channel gains may result in some  users served by the data throughput less than their demands. The network may improve the service for these unsatisfied users by efficiently using the power leftover, which has been allocated to the satisfied users with a higher served data throughput than requested.}
	\end{assumption}
	 \textcolor{black}{Assumption~\ref{AssumpPrior} justifies a priority on the data throughput requirements of the scheduled users, which is effectively attained by $\Qcal$. The limited transmit power should be, therefore, used in a strategy to maximize the service demands for all the users  rather than each single entity. The remaining power, if still available, will be exploited to maximize the sum data throughput.  Thanks to the Perron-Frobenius theorem \cite{pillai2005perron},  the conditions needed to all the   users with their data throughput satisfactions are given as in Theorem~\ref{Cardinality}.}
	\begin{theorem} \label{Cardinality}
		\textcolor{black}{When $\UEk$ requires a non-zero data throughput ($\xi_k > 0$), then the network can serve all the $K$ users with at least their individual data requirements if
		\fontsize{10}{10}{\begin{align}
				&\lambda(\mathbf{R} \mathbf{Q}) < 1, \label{eq:rho}\\
				& \mathbf{1}_K^T (\mathbf{I}_K - \mathbf{R}\mathbf{Q})^{-1} \pmb{\nu} \leq P_{\max}, \label{eq:PowerConstr}
		\end{align}}
		\hspace*{-3pt}in which $\pmb{\nu} = [\nu_1, \ldots, \nu_K]^T \in \mathbb{R}_{+}^K$, $\nu_k = \alpha_k \sigma^2 / ( (\alpha_k +1 ) |\mathbf{h}_k^2 \mathbf{w}_k|^2 )$, and $\alpha_k = 2^{\xi_k/B} -1, \forall k \in \Kcal$; $\mathbf{R} \in  \mathbb{R}^{K \times K}$ is a matrix whose  $(k,k')-$th element given as $[\mathbf{R}]_{kk'} = \frac{\alpha_k}{(\alpha_k +1)|\mathbf{h}_k^H \mathbf{w}_k|^2}$ if $k=k'$. Otherwise, $[\mathbf{R}]_{kk'} = 0$.  
		The $(k,k')$-th element of $ \mathbf{Q} \in  \mathbb{R}^{K \times K}$ is 
		$[\mathbf{Q}]_{kk'} = | \mathbf{h}_k^H \mathbf{w}_{k'} |^2$.
		In \eqref{eq:rho}, $\lambda(\mathbf{R} \mathbf{Q}) = \max \{ |\lambda_1|, \ldots, |\lambda_K| \}$ is the spectral radius of $\mathbf{R} \mathbf{Q}$ with the eigenvalues $\lambda_1, \ldots, \lambda_K$.}
	\end{theorem}
	\begin{proof}
	\textcolor{black}{The proof is adopted from the previous works on single-wireless links \cite{pillai2005perron} to the MB-HTS systems. The detailed proof is omitted due to space limitations.}
	\end{proof}
	\textcolor{black}{Theorem~\ref{Cardinality} offers the criteria for the network to serve all the $K$  users with the data throughput requirements, while still optimizing $f_0 (\{ p_{k'}\})$. Different from the previous works,  \eqref{eq:rho} and \eqref{eq:PowerConstr} point out the existed unique power solution for a precoded satellite system as a multi-variate function of many variables comprising the channel information, the precoding method, the noise power, the data throughput requirements, and the transmit power. The total power required to satisfy the service demands is bounded from below  as in Corollary~\ref{CorrNoSatisfied}.}
	\begin{corollary}\label{CorrNoSatisfied}
		\textcolor{black}{For a given set of data throughput requirements from the $K$ users, the total transmit power consumption is bounded by 
		$\sum\nolimits_{k \in \Kcal} p_k \geq \mathbf{1}_K^T \pmb{\nu} / \| \mathbf{I}_K - \mathbf{R} \mathbf{Q} \|_2.$}
	\end{corollary}
	\begin{proof}
		\textcolor{black}{The proof derives a lower bound of  \eqref{eq:PowerConstr} and is omitted due to space limitations.}
	\end{proof}
	\begin{algorithm}[t]
		\begin{algorithmic}[1]\fontsize{9}{9}\selectfont
			\protect\caption{\textcolor{black}{A sub-optimal solution to problem \eqref{probGlobal}}}
			\label{alg:glob_alg}
			\global\long\def\algorithmicrequire{\textbf{INPUT:}}
			\REQUIRE \textcolor{black}{Propagation channels $\{\bh_{k'}\}$; Maximum transmit power $P_{\max}$;  Data throughput requirement (QoS) set $\{\xi_{k'}\}$.}
			\STATE \textcolor{black}{Formulate the precoding vectors $\{\bw_{k'}\}$ based on $\{ \mathbf{h}_{k'} \}$.}
			\STATE \textcolor{black}{Define the matrices $\mathbf{R}, \mathbf{Q},$ and vector $\pmb{\nu}$.} 
			\IF {\textcolor{black}{ \eqref{eq:rho} and \eqref{eq:PowerConstr} are satisfied}}
			\STATE \textcolor{black}{Update the demand-satisfied set $\Qcal^\ast = \Kcal$ and solve \eqref{RelatedWork} to attain  $\{ p_{k'}^{\ast} \}$.}
			\ELSE
			\STATE \textcolor{black}{Solve \eqref{RelatedWorkSumRate} to attain $\{ p_{k'}^{\ast,(0)} \}$ and  $\Qcal^{\ast,(0)}$.}
			\STATE \textcolor{black}{Initialize setting $\delta = |\Qcal^{\ast, (0)}|$ and $n=0$.}
			\WHILE {$\delta \neq 0$} 
			\STATE \textcolor{black}{Update iteration index $n = n+1$.}	
			\STATE \textcolor{black}{Solve \eqref{alg_SRM:phase23} to attain $\{ p_{k'}^{\ast, (n)} \}$ and then update $\Qcal^{\ast, (n)}$.} 
			\STATE \textcolor{black}{Recalculate $\delta = |\Qcal^{\ast, (n)}| - |\Qcal^{\ast, (n-1)}|$.} 
			\ENDWHILE
			\ENDIF
			\global\long\def\algorithmicrequire{\textbf{OUTPUT:}}
			\REQUIRE  \textcolor{black}{The demand-satisfied set $\Qcal^{\ast} = \Qcal^{\ast, (n)}$ and the optimized transmit powers $ \{p_{k'}^\ast \} = \{p_{k'}^{\ast, (n)} \}$.}
		\end{algorithmic}
	\end{algorithm}
	\textcolor{black}{The lower bound in Corollary~\ref{CorrNoSatisfied} demonstrates that the total power consumption is always positive if each user request a non-zero data throughput because of the interference and noise. In addition, it shows the contributions of the precoding method. A proper exploitation can effectively reduce the interference among the  users to obtain the large spectral norm of $\mathbf{I}_K - \mathbf{R} \mathbf{Q}$. For balancing between the computational complexity and the system performance, a linear precoding method can be deployed, e.g.,
	\begin{equation}
		\fontsize{10}{10}{\bW = \begin{cases}
				\bH (\bH^H\bH )^{-1}, & \mbox{For ZF},\\
				\bH \left( \bH^H\bH + \frac{K\sigma^2}{P_{\max}} \bI_K \right)^{-1}, & \mbox{For RZF},
		\end{cases}}
	\end{equation}
	then the precoding vector used for  $\UEk$ is formulated as $\mathbf{w}_k = [\bW]_k / \| \bW\|_k$ with $[\bW]_k$ being the $k$-th column of matrix $\bW$. Conditioned on the set of precoding vectors $\{\bw_{k'}\}$ defined from the spropagation channels $\{\bh_{k'}\}$, \eqref{eq:rho} and \eqref{eq:PowerConstr} lead to two possible cases: If these conditions hold, then $R_k(\{p_{k'}\}) \geq  \xi_k, \forall k\in\Kcal$, and $\Qcal = \Kcal$. Based on Theorem~\ref{Cardinality} and Assumption~\ref{AssumpPrior}, problems~\eqref{probGlobalv1} and \eqref{RelatedWork} are equivalent to each other. Accordingly, it provides a nonempty feasible set and a network without congestion. Another case is  as one of these conditions does not hold, resulting in one   user does not satisfy its data throughput requirement, and therefore the congestion happens. A particular deal should be considered to tackle this issue once considers the traditional sum data throughput optimization \eqref{RelatedWork} because of an empty feasible set. Nevertheless, it is not such the situation for our proposed problem in \eqref{probGlobalv1}. We emphasize that the first case optimizes the sum data throughput with the demand-based constraints of all the $K$   users and a finite transmit power level. In order for the network to solve problem~\eqref{probGlobalv1}, the priority is to maximize the number of demand-satisfied users by maximizing the cardinality of $\Qcal$ as}
	\begin{subequations} \label{alg_SRM:phase21}
		\begin{alignat}{2}
			&\underset{\{ p_{k'} \in \mathbb{R}_{+} \}}{\mathrm{maximize}} \quad   |\Qcal|\\
			&\mbox{subject to} \quad \  \eqref{probGlobalb}, \eqref{probGlobalc}, \eqref{probGlobald}.
		\end{alignat}
	\end{subequations}
	\textcolor{black}{Because \eqref{alg_SRM:phase21} is a non-convex problem, the global optimum is nontrivial to obtain. We, therefore, propose an iterative algorithm to overcome this issue with a good sub-optimal solution. First, we solve the sum data throughput optimization problem without the demand-based constraints in \eqref{RelatedWorkSumRate} to attain an initial transmit power coefficients $\{ p_{k'}^{\ast, (0)}\}$, which is later utilized to formulate the initial demand-satisfied set $\Qcal^{\ast, (0)}$ as
	$\fontsize{10}{10}{\Qcal^{\ast, (0)} = \big\{ k \big| R_k \big( \{ p_{k'}^{\ast, (0)}\} \big) \geq \xi_k, k \in \Kcal \big\}. }$
	Notably, the users typically satisfying their data throughput  requirements are ones with good channel conditions and these suffering less from interference contributing significantly to the objective function of  \eqref{RelatedWorkSumRate}. 
	Next, we construct an approach, which can expand  $\Qcal$ iteratively. The key idea is that the demand-satisfied users in $\Qcal$ are only offered service equal to the data throughput requirements, all the remaining transmit power will be reallocated to the other users to upgrade their service. Consequently, there is an opportunity for these users to join $\Qcal$. We focus on the following problem at iteration~$n$ as
	\begin{subequations} \label{alg_SRM:phase23}
		\begin{alignat}{2}
			&\underset{\{ p_{k'}^{(n)} \in \mathbb{R}_{+} \}}{\mathrm{maximize}} &\quad& \sum\nolimits_{k\in\Kcal}R_{k} \big(\{p_{k'}^{(n)}\} \big) \label{alg_SRM:phase23a}\\
			&\mbox{subject to} &\quad& R_k(\{ p_{k'}^{(n)} \} ) = \xi_k, \forall k \in\Qcal^{\ast, (n-1)}, \label{alg_SRM:phase23b}\\
			&&& \sum\nolimits_{k\in\Kcal}p_k^{(n)} \leq P_{\max}, \label{alg_SRM:phase23c}
		\end{alignat}
	\end{subequations}
	with the optimal transmit power $\{ p_{k'}^{\ast,(n)} \}$. The constraints \eqref{alg_SRM:phase23b} makes the GEO satellite  serve the demand-satisfied users  in $\Qcal$ with only their data throughput requirements.  
	The remaining transmit power will be reallocated to the users with weak channel conditions to boost an opportunity that they can be potential candidates of $\Qcal$. If users are served by at least their data throughput requirements at iteration~$n$, they will be joined to $\Qcal$ as
	$\fontsize{9}{9}{\Qcal^{\ast, (n)} =  \big\{ k \big| R_k \big(\{ p_{k'}^{\ast, (n)} \big) \geq \xi_k, k \in \Kcal \big\}. }$
Following, the iteration index is updated as $n=n+1$. 
We should note that our proposal maximizes the number of demand-satisfied  users in each iteration with the objective to maximize the sum data throughput of all the $K$ users. The proposed design is given in Algorithm \ref{alg:glob_alg} and its convergence given in Theorem~\ref{theorem:Coverge}.}
		\begin{figure*}[ht]
	\centering
	\includegraphics[trim=0.0cm 0cm 4.8cm 0.0cm, clip=true, width=7.2in, width= 0.92 \textwidth]{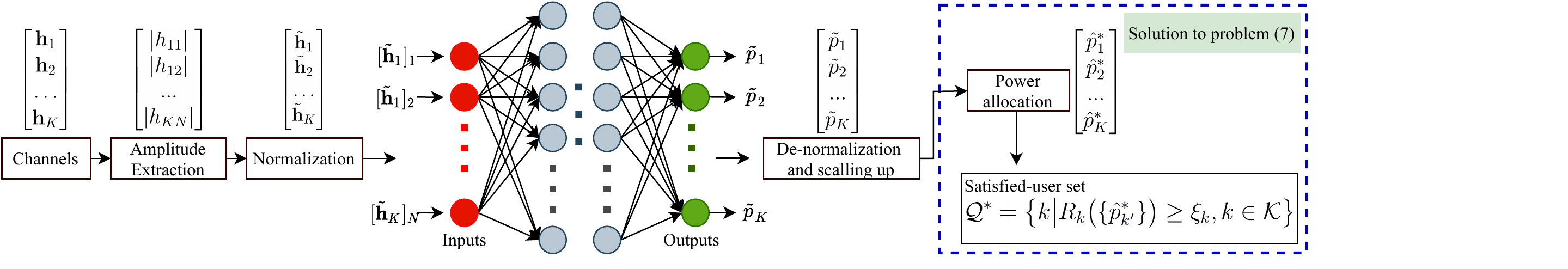}
	\caption{\textcolor{black}{The proposed fully-connected  neural network architecture to learn and jointly predict the transmit powers and demand-satisfied set.}}
	\label{fig:DNN_model}
\end{figure*}
	\begin{theorem} \label{theorem:Coverge}
		\textcolor{black}{As the $K$ users are served under their data throughput requirements with a given transmit power level $P_{\max}$ at the satellite,  the convergence properties are}
		\fontsize{10}{10}{\begin{align}
				\ldots & \geq |\Qcal^{\ast,(n)}| \geq |\Qcal^{\ast,(n-1)}| \geq \ldots \geq |\Qcal^{\ast,(0)}|, \label{eq:QSeries}\\
				\ldots & \leq \sum\nolimits_{k \in \Kcal} R_k \big( \{ p_{k'}^{\ast,(n)} \} \big) \leq \sum\nolimits_{k \in \Kcal} R_k \big( \{ p_{k'}^{\ast,(n-1)} \} \big) \nonumber \\
				&\leq  \ldots \leq \sum\nolimits_{k \in \Kcal} R_k \big( \{ p_{k'}^{\ast,(0)} \} \big). \label{eq:RSeries}
		\end{align}}
	\end{theorem}
	\begin{proof}
	\textcolor{black}{The proof is to testify the monotonic property of the sum data throughput and the cardinality of  demand-satisfied set and is omitted due to space limitations.}
	\end{proof}
	\textcolor{black}{We have pointed out in Theorem~\ref{theorem:Coverge} that an improvement of the demand-satisfied set after each iteration can be attained by sacrificing the sum-data throughput based on the $\pmb{\epsilon}$-properly Pareto optimality presented in Section~\ref{sec:MultiObject}. Once the congestion happens, the $K$ users must be chopped into the two sets: the demand-satisfied set $\Qcal$ including the users with at least their data throughput demands, and the demand-unsatisfied set $\Kcal \setminus \Qcal$ with the remaining users with the data throughput below their requirements.}
	\begin{remark}
	   \textcolor{black}{As a priority,	Algorithm~\ref{alg:glob_alg} first maintains the data throughput requirements for the users under a finite transmit power budget. Alternately, it strategically allocates the transmit power to achieve the maximum number of demand-satisfied user, and then the sum data throughput maximization is considered. Our framework gives room for the decision maker to design specific approaches in order to solve the sum data throughput optimization problems~\eqref{RelatedWorkSumRate} and \eqref{alg_SRM:phase23}. However, this paper employs the semi-closed form solution of the water-filling method to acquire a cost-effective design.}
	\end{remark}
	\subsection{Data-Driven Solution} \label{Sec:DataDriven}
	\textcolor{black}{This subsection utilizes a neural network  to learn the characteristics of Algorithm~\ref{alg:glob_alg} and predict its solution with low complexity. 
The optimized power solution provided by Algorithm~\ref{alg:glob_alg} is assumed to be available and the series of continuous mappings are formulated as follows
	\begin{align} 
		\mathbf{w}_{\ell} =& \tilde{\mathbf{f}}_\ell ( \{ \mathbf{h}_k \} ),\  \forall \ell \in \Kcal, \label{eq:well}\\
		\mu_{kl}  =& |\bh_k^H\bw_\ell|^2,\  \forall k,\ell \in \Kcal, \label{eq:M1}\\
		\alpha_{k}^{\ast}  =& \frac{p_k^\ast  \mu_{kk} }{\sum_{\ell\in\Kcal\backslash \{k\}}p_{\ell}^\ast  \mu_{kl}  +\sigma^2} ,\  \forall k \in \Kcal, \label{eq:M2}\\
		p_k^{\ast}  =&  \alpha_{k}^{\ast}  \frac{\sigma^2}{\mu_{kk}}  + \alpha_{k}^{\ast} 	\sum\nolimits_{\ell \in \mathcal{K} \setminus \{k\} } p_{\ell }^\ast \frac{\mu_{kl} }{\mu_{kk} },\forall  k \in \Kcal,\label{eq:M3}
	\end{align}
where $\tilde{\mathbf{f}}_\ell ( \{ \mathbf{h}_k \} )$ denotes a function that defines a precoding vector user~$\ell$ from the channel information. After the process in \eqref{eq:well}, the  $K$ precoding vectors are defined, which are used as the input to compute the gains or mutual interference in the mapping \eqref{eq:M1} corresponding $k = l$ or $k \neq l$, respectively. 	The continuous mapping \eqref{eq:M2} computes the SINR  for  user~$k, \forall k$. We notice that the optimized demand-satisfied set $\Qcal^\ast$ should be explicitly formulated by $\{ p_k^\ast \}$ together with $\{ \alpha_k^\ast \}$ in \eqref{eq:M2}. From the series of continuous mappings in \eqref{eq:well}--\eqref{eq:M3}, a low-cost  neural network can be constructed by a limited number of neurons for our problem. Furthermore, \eqref{eq:M3} indicates a method to update the transmit power of $\UEk$ in relation to the offered data throughput and the power allocated to the other users.}
	\begin{lemma} \label{lemmaExistNeural}
		 \textcolor{black}{Algorithm~\ref{alg:glob_alg}  yields the optimized transmit power coefficients, which can be characterized by the process
		$\{ p_k \} = \mathcal{F} ( \{ \mathbf{h}_k \} )$,
		in which $\mathcal{F} ( \{ \mathbf{h}_k \} )$ denotes the series of the continuous mappings in \eqref{eq:well}--\eqref{eq:M3}. It means that  a neural network can be trained to learn the mapping $\mathcal{F} ( \{ \mathbf{h}_k \} )$.}
	\end{lemma}
	\begin{proof}
		\textcolor{black}{The proof is since $\Qcal$ can be computed by  $\{ p_{k'} \}$ and is omitted due to space limitations.}
	\end{proof}
	\textcolor{black}{Lemma~\ref{lemmaExistNeural} unveils that  neural networks with supervised learning only attract insightful information of the propagation channels to learn the features of $\mathcal{F}(\{ \mathbf{h}_k \})$, and then predict the transmit powers with low complexity thanks to the fact that  $\Qcal$ is expressed in \eqref{eq:M2}. As one of the main contributions, we only exploit the channel gains, i.e., magnitudes only, to design a fully-connected deep neural network based on \eqref{eq:M1} for given precoding vectors as illustrated in Fig.~\ref{fig:DNN_model}.}
	
		\begin{figure*}[t]
		\begin{minipage}{0.325\textwidth}
			\centering
			\includegraphics[width= 1.1 \textwidth]{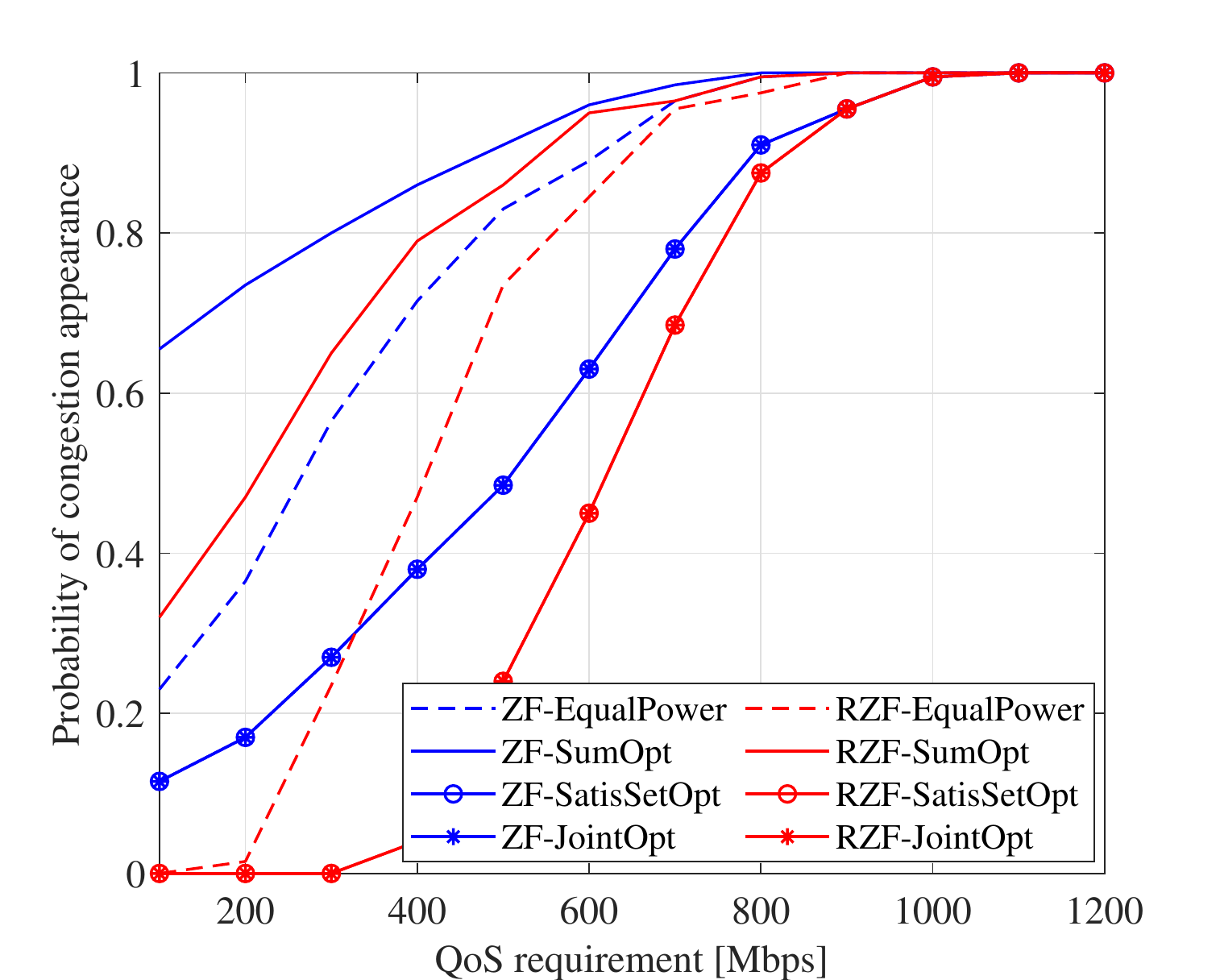} \\
			$(a)$
		\end{minipage}
		\hfill
		\begin{minipage}{0.325\textwidth}
			\centering
			\includegraphics[width= 1.1 \textwidth]{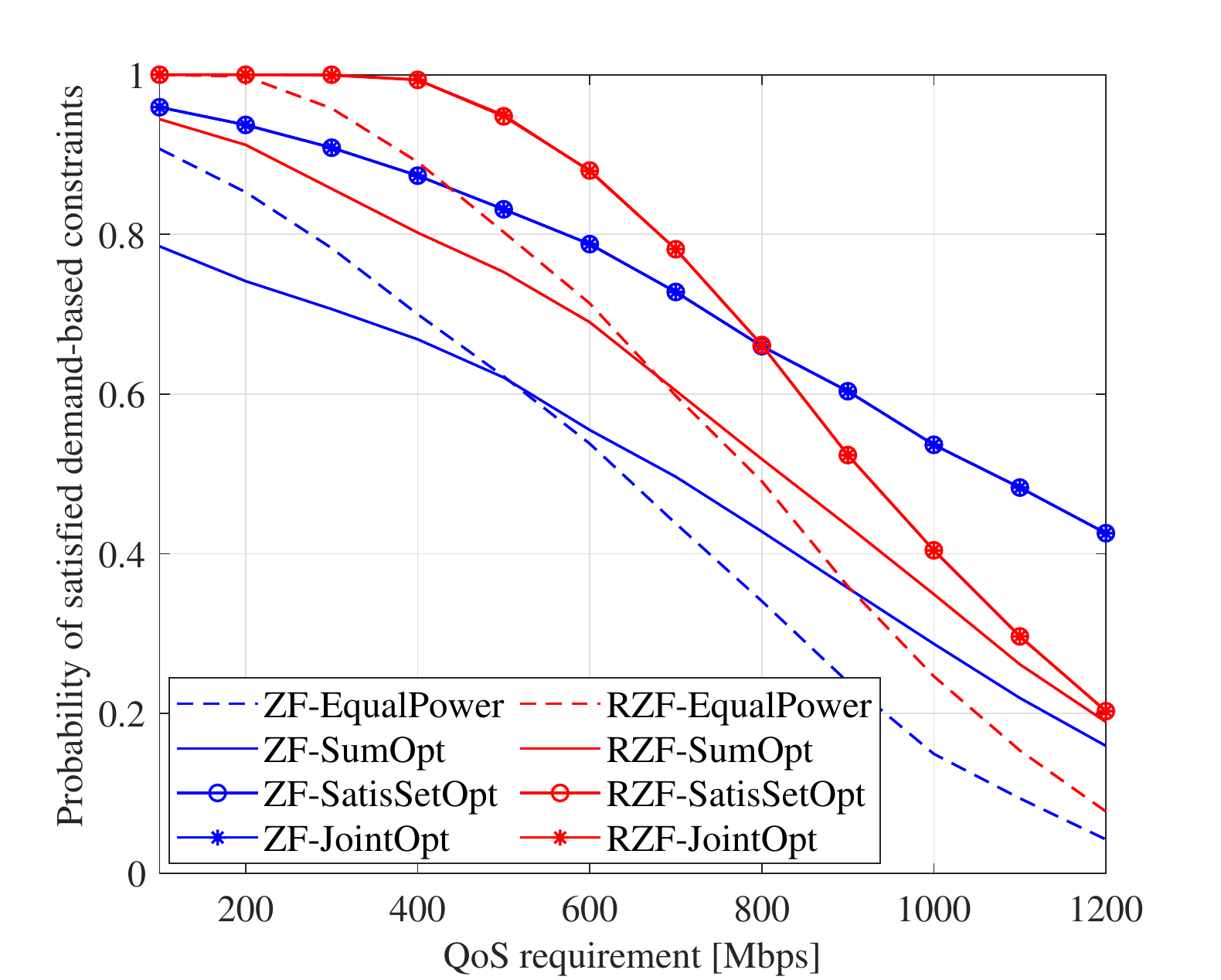} \\
			$(b)$
		\end{minipage}
		\hfill
		\begin{minipage}{0.325\textwidth}
			\centering
			\includegraphics[width= 1.1 \textwidth]{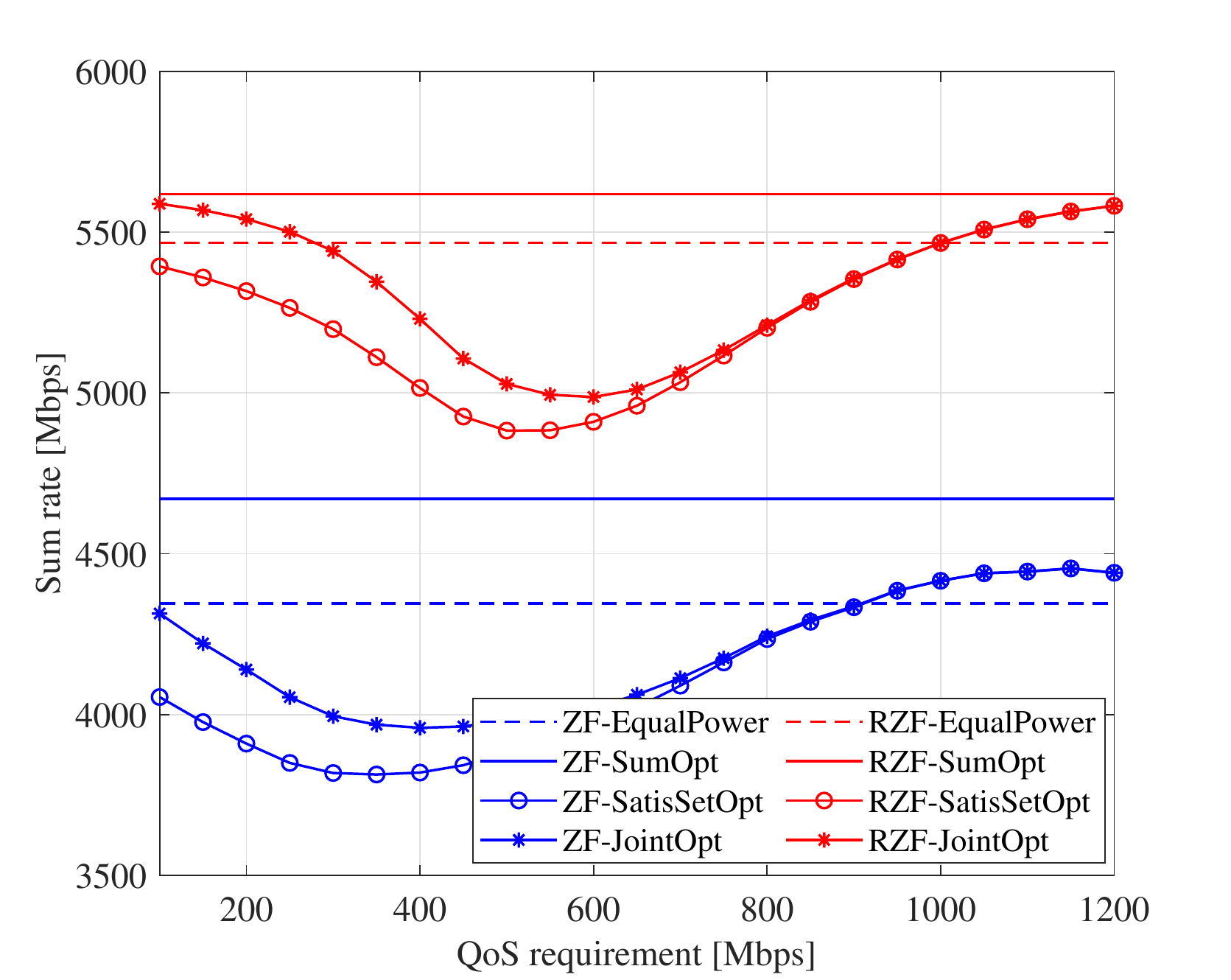} \\
			$(c)$
		\end{minipage}
		\caption{\textcolor{black}{The system performance:
				$(a)$ plots The probability of congestion appearance vs the data throughput requirement; $(b)$ plots the probability of demand-satisfied users vs the data throughput requirement; and $(c)$ plots the sum data throughput vs the data throughput requirement.}} \label{Fig:SysPer}
	\end{figure*}
	\textcolor{black}{We define the vector $\tilde{\bh}_{k} = [ |h_{k1}|, \ldots,  |h_{kN}|]^T \in \mathbb{R}_{+}^N$ that includes the channel gains for $\UEk$, then  all the required channel gains needed for the input of the neural network are collected into a vector as $\mathbf{x} =[\tilde{\bh}_{1}^T, \dots,\tilde{\bh}_{K}^T]^T \in\mathbb{R}_+^{KN}$. We stress that the channel gains may be sometimes extremely small due to deep fading; thus, they are normalized to compensate the fluctuations and randomness from the radio environment. Mathematically, the normalized vector $\mathbf{x}_{\mathsf{in}} \in \mathbb{R}^{KN}$ is derived from $\mathbf{x}$ as
		\begin{equation} \label{eq:xinm}
	[\mathbf{x}_{\mathsf{in}}]_m = ([\mathbf{x}]_m -[\mathbf{x}_{\min}]_m)/( [\mathbf{x}_{\max}]_m - [\mathbf{x}_{\min}]_m),
		\end{equation}
	where $[\mathbf{x}]_m$ denotes the $m$-th entry of vector $\mathbf{x}$; $\mathbf{x}_{\min}, \mathbf{x}_{\max} \in \mathbb{R}_{+}^{KN}$ with
	$[\mathbf{x}_{\min}]_m = \min \, \{ [\mathbf{x}]_m \} \mbox{ and } [\mathbf{x}_{\max}]_m = \max \, \{ [\mathbf{x}]_m \},$ 
	in which $\{ [\mathbf{x}]_m \}$ contains the different realizations of $[\mathbf{x}]_m$. In order to create a compact set, the channel gains and the transmit powers are normalized. After that, the normalized data $\mathbf{x}_{\mathsf{in}}$ is the input of the neural network to optimize the values of biases and weights.  
	 Specifically, as $\mathbf{x}_{uv}$ represents the input vector of the $u$-th neuron of the $v$-th hidden layer, the output  is given as 
	 \begin{equation}\label{}
	 	y_{uv}= f_{uv}(\mathbf{w}_{uv}^T \mathbf{x}_{uv} + b_{uv})
	 \end{equation}
 with $\mathbf{w}_{uv}$ and $b_{uv}$ being the corresponding weight and bias; $f_{uv}(\cdot)$ denote the activation function. Let us denote 
	the output of the neural network as $\tilde{\mathbf{p}} \in\mathbb{R}_+^K$.
	In  the testing phase,	the predicted transmit power vector $\hat{\mathbf{p}}  \in \mathbb{R}_+^K$ is denormalized as
		\begin{equation}
	[\hat{\mathbf{p}}]_k = [\tilde{\mathbf{p}}]_k \left( [\tilde{\mathbf{p}}_{\max}]_k - [\tilde{\mathbf{p}}_{\min}]_k \right)  + [\tilde{\mathbf{p}}_{\min}]_k,
		\end{equation}
	where $[\cdot]_k$ is its $k$-th element, whilst $[\tilde{\mathbf{p}}_{\max}]_k$ and $[\tilde{\mathbf{p}}_{\min}]_k$ represent the maximum and minimum transmit powers for $\UEk$ in the data set. To ensure the constraint \eqref{probGlobalc}, we use the following mapping
		\begin{equation}
	[\hat{\mathbf{p}}^\ast]_k = P_{\max} [\hat{\mathbf{p}}]_k  \big/ \sum\nolimits_{k' \in \Kcal} [\hat{\mathbf{p}}]_{k'},
		\end{equation}
	and then $\sum\nolimits_{k \in \Kcal}  [\hat{\mathbf{p}}^\ast]_k = P_{\max}$ is related to the full power consumption \cite{Bjon13:tcit}.
	The mean squared error (MSE) criteria is exploited to define the loss function as
		\begin{equation}\label{loss_function}
	\mathcal{L}^\text{MSE}( {\Theta}) = \mathbb{E} \{ \|\tilde{\mathbf{p}} - \tilde{\mathbf{p}}^{\ast}\|^2_2 \},
		\end{equation}
	where ${\Theta}$ includes all the weights and biases; $\tilde{\mathbf{p}}^\ast$ contains the transmit powers attained from Algorithm~\ref{alg:glob_alg} after normalization.} 
	\section{Numerical Results} \label{Sec:Results}
	A GEO satellite network with $N=7$ beams  serving $K=7$ users in each coherence interval with either the ZF or RZF precoding method is considered for simulations. The parameters associated with the satellite and the beam radiation patterns are provided by European Space Agency (ESA)  in the context of \cite{ESA}. In detail, the radiation patterns are based on a Defocused Phased Array-Fed Reflector (PAFR), with reflector size of $2.2$m and an array diameter of roughly $1.2$m. The antenna array before the reflector is a circular array with the space of $2\times$ carrier wavelength  and $511$ elements.The satellite location is at $13^\circ$~E, and the system operates at Ka band, for which the carrier frequency is $20$~[GHz] \cite{CGD}. The system bandwidth is $500$~[MHz] and the satellite height is $35,786$~[km].  The receive antenna diameter is $0.6$~m and the noise power is about $-118.3$~dB.  For the data transmission, the maximum transmit power is $P_{\max}  = 23.37$~dBW.  For the data-driven solution, we build a fully-connected neural network with two hidden layers including $128$ and $64$ neurons, respectively. The $25000$ trials of various user locations are gathered for the training phase, while  the $10000$ trials are used for the testing phase. By employing the water-filling technique to attain a semi-closed form solution in each iteration to the transmit powers, the following benchmarks are considered:
	\begin{itemize}
		\item[$i)$] \textcolor{black}{\textit{Joint sum data throughput and demand-satisfied set maximization} (JointOpt) is given in Algorithm~\ref{alg:glob_alg}. We apply this algorithm  for the ZF and RZF beamforming methods.} 
		\item[$ii)$] \textcolor{black}{\textit{Demand-satisfied set maximization} (SatisSetOpt) is a relaxation of JointOpt  guaranteeing users' demands only. When all the users are served with at least their requested data throughput, the remaining transmit power budget is equally allocated to every user.} 
		\item[$iii)$] \textcolor{black}{\textit{Sum data throughput maximization} (SumOpt) has been recently proposed in \cite{lu2019robust}, which only focuses on  the sum data throughput maximization. Consequently, users with weak channel conditions may be ignored from service.} 
		\item[$iv)$] \textcolor{black}{\textit{Equal transmit power allocation} (EqualPower) is a baseline to verify the merits of power control and demand-satisfied maximization \cite{trinh2021user, Krivochiza:21:access}. A transmit power amount of $14.92$~dB is allocated to each user without a maintenance on users'  demands.}  
	\end{itemize}
	\textcolor{black}{Fig.~\ref{Fig:SysPer}(a) shows the probability in which the congestion appears. Specifically, it  is defined by the trials when the system cannot serve all users in the coherence interval with their data throughput requirements. If the data throughput requirement increases, our algorithm offers the lowest congestion probability for both the precoding methods, i.e., the RZF and ZF, especially at a low data throughput value. Since SumOpt only concentrates on the total  sum data throughput, it gives the highest congestion. The reason is that users with weak channel gains get less transmit power, and therefore there is no data throughput guarantee for them. Fig.~\ref{Fig:SysPer}(b) plots the probability of the demand-based constraint satisfaction, which is $\mathbb{E}\{ |\Qcal| \}/K$.  As the data throughput requirement grows up, the satisfaction degrades because the system with a finite transmit power level faces challenges in guaranteeing the requests from multiple users. Meanwhile, Fig. \ref{Fig:SysPer}(c) visualizes the scarification of sum data throughput to have more demand-satisfied users. EqualPower and SumOpt provide the constant sum rate since these benchmarks allocate the transmit power without any guarantee.} 
	\begin{table}[t]
		\caption{The  run time (milliseconds) and sum data throughput of the model-based and data-driven solutions}
		\label{runtime}
		\centering
		\begin{tabular}{|l|c|c|c|c|}
			\hline
			& \begin{tabular}[c]{@{}l@{}}Data Thro. \\ require.\\ $\mbox{[Mbps]}$\end{tabular} & \begin{tabular}[c]{@{}l@{}}Time \\ {[}ms{]}\end{tabular} & \begin{tabular}[c]{@{}l@{}}Sum  \\  {[}Mbps{]}\end{tabular} & \begin{tabular}[c]{@{}l@{}}Percentage \\ of satis.\\ {[}\%{]}\end{tabular} \\ 
			\hline
			\multirow{2}{*}{Model-based (ZF)}& $\xi_k = 250$ & $17.38$  & $4054$ & $92.14$ \\ 
			\cline{2-5} 
			& $\xi_k = 500$   & $19.7$ & $3984$ & $83.14$  \\ 
			\hline
			\multirow{2}{*}{Model-based (RZF)} & $\xi_k = 250$ & $19.26$& $5542$ & $99.86$ \\ 
			\cline{2-5} 
			& $\xi_k = 500$  & $26.7$ & $5077$ & $94.69$  \\ 
			\hline
			\multirow{2}{*}{Data-driven (ZF)}  & $\xi_k = 250$ & $2.0$ & $4260$& $82.38$ \\ \cline{2-5} 
			& $\xi_k = 500$ & $2.1$ & $4195$ & $65.88$ \\ 
			\hline
			\multirow{2}{*}{Data-driven (RZF)} & $\xi_k = 250$ & $1.3$ & $5545$ & $98.47$ \\ \cline{2-5} 
			& $\xi_k = 500$ & $1.8$ & $5124$& $87.15$ \\ 
			\hline
		\end{tabular}
	\end{table}
\textcolor{black}{In addition, the sum data throughput and run time of our proposed solutions are given in Table~\ref{runtime}.  {\color{black} It is noted that the performance of the data-driven solution for ZF precoding might  be lower in comparison with others (e.g., RZF) because the continuous mappings in \eqref{eq:well}--\eqref{eq:M3} may  not be isomorphisms since the codomains are non-smooth functions, especially for the achievable rates in \eqref{eq:M2}. The fact manifests difficulties in training and predicting the joint power allocation and satisfied-user set optimization.} Despite a slightly higher run time as the data throughput demands  increase, our solutions consume time in milliseconds (ms) that is  very fast for GEO satellite systems. Properly defining the continuous mappings makes the data-driven solution have lower run time  $15\times$ than the model-based solution.} 
	
	\section{Conclusions}\label{sec:conclusion}
	\textcolor{black}{The congestion control has been studied under the service-based constraints for multi-beam multi-user GEO satellite systems. By utilizing the multi-objective optimization theory, we have demonstrated that  the sum data throughput and satisfied-user set with all the propagation channel conditions can be jointly optimized.  With respect to guarantee the data throughput requirements as the priority, we have utilized the model-based solution to design an algorithm effectively operating in all the domains, especially the infeasibility caused by the limited transmit power and the data throughput demands. In addition, the time consumption by deploying a fully-connected neural network is much less than  $10$~ms making a step toward practical real-time power control as well as satisfied-user maintenance in satellite networks.}
	\bibliographystyle{IEEEtran}
	\balance
	\bibliography{Journal}

\begin{thebibliography}{10}
\providecommand{\url}[1]{#1}
\csname url@samestyle\endcsname
\providecommand{\newblock}{\relax}
\providecommand{\bibinfo}[2]{#2}
\providecommand{\BIBentrySTDinterwordspacing}{\spaceskip=0pt\relax}
\providecommand{\BIBentryALTinterwordstretchfactor}{4}
\providecommand{\BIBentryALTinterwordspacing}{\spaceskip=\fontdimen2\font plus
\BIBentryALTinterwordstretchfactor\fontdimen3\font minus
  \fontdimen4\font\relax}
\providecommand{\BIBforeignlanguage}[2]{{%
\expandafter\ifx\csname l@#1\endcsname\relax
\typeout{** WARNING: IEEEtran.bst: No hyphenation pattern has been}%
\typeout{** loaded for the language `#1'. Using the pattern for}%
\typeout{** the default language instead.}%
\else
\language=\csname l@#1\endcsname
\fi
#2}}
\providecommand{\BIBdecl}{\relax}
\BIBdecl

\bibitem{Kodheli:21:tut}
O.~Kodheli, E.~Lagunas, N.~Maturo, S.~K. Sharma, B.~Shankar, J.~F.~M. Montoya,
  J.~C.~M. Duncan, D.~Spano, S.~Chatzinotas, S.~Kisseleff \emph{et~al.},
  ``Satellite communications in the new space era: A survey and future
  challenges,'' \emph{IEEE Communications Surveys \& Tutorials}, vol.~23,
  no.~1, pp. 70--109, 2020.

\bibitem{ZhangAsilomar2020}
Y.~D. Zhang and K.~D. Pham, ``Joint precoding and scheduling optimization in
  downlink multicell satellite communications,'' in \emph{2020 54th Asilomar
  Conference on Signals, Systems, and Computers}.\hskip 1em plus 0.5em minus
  0.4em\relax IEEE, 2020, pp. 480--484.

\bibitem{Zheng:twc:12}
G.~Zheng, S.~Chatzinotas, and B.~Ottersten, ``Generic optimization of linear
  precoding in multibeam satellite systems,'' \emph{IEEE Transactions on
  Wireless Communications}, vol.~11, no.~6, pp. 2308--2320, 2012.

\bibitem{trinh2021user}
T.~Van~Chien, E.~Lagunas, T.~H. Ta, S.~Chatzinotas, and B.~Ottersten, ``User
  scheduling for precoded satellite systems with individual quality of service
  constraints,'' in \emph{2021 IEEE 32nd Annual International Symposium on
  Personal, Indoor and Mobile Radio Communications (PIMRC)}.\hskip 1em plus
  0.5em minus 0.4em\relax IEEE, 2021, pp. 1042--1047.

\bibitem{Lei:Access:20}
L.~Lei, E.~Lagunas, Y.~Yuan, M.~G. Kibria, S.~Chatzinotas, and B.~Ottersten,
  ``Beam illumination pattern design in satellite networks: Learning and
  optimization for efficient beam hopping,'' \emph{IEEE Access}, vol.~8, pp.
  136\,655--136\,667, 2020.

\bibitem{Chien19:TCOM}
T.~Van~Chien, T.~N. Canh, E.~Bj{\"o}rnson, and E.~G. Larsson, ``Power control
  in cellular massive mimo with varying user activity: A deep learning
  solution,'' \emph{IEEE Transactions on Wireless Communications}, vol.~19,
  no.~9, pp. 5732--5748, 2020.

\bibitem{Chen:cof:20}
R.~Chen, X.~Hu, X.~Li, and W.~Wang, ``Optimum power allocation based on traffic
  matching service for multi-beam satellite system,'' in \emph{2020 5th
  International Conference on Computer and Communication Systems
  (ICCCS)}.\hskip 1em plus 0.5em minus 0.4em\relax IEEE, 2020, pp. 655--659.

\bibitem{ESA}
\BIBentryALTinterwordspacing
ESA, ``{SATellite Network of EXperts ({SATNEX}) IV}.'' [Online]. Available:
  \url{https://satnex4.org/}
\BIBentrySTDinterwordspacing

\bibitem{lu2019robust}
W.~Lu, K.~An, and T.~Liang, ``Robust beamforming design for sum secrecy rate
  maximization in multibeam satellite systems,'' \emph{IEEE Transactions on
  Aerospace and Electronic Systems}, vol.~55, no.~3, pp. 1568--1572, 2019.

\bibitem{Krivochiza:21:access}
J.~Krivochiza, J.~C.~M. Duncan, J.~Querol, N.~Maturo, L.~M. Marrero,
  S.~Andrenacci, J.~Krause, and S.~Chatzinotas, ``End-to-end precoding
  validation over a live {GEO} satellite forward link,'' \emph{arXiv preprint
  arXiv:2103.11760}, 2021.

\bibitem{Bjon13:tcit}
E.~Björnson and E.~Jorswieck, ``Optimal resource allocation in coordinated
  multi-cell systems,'' \emph{Found. Trends Commun. Inf. Theory}, vol.~9, no.
  2-3, pp. 113--381, 2013.

\bibitem{aravanis2015power}
A.~I. Aravanis, B.~S. MR, P.-D. Arapoglou, G.~Danoy, P.~G. Cottis, and
  B.~Ottersten, ``Power allocation in multibeam satellite systems: A two-stage
  multi-objective optimization,'' \emph{IEEE Transactions on Wireless
  Communications}, vol.~14, no.~6, pp. 3171--3182, 2015.

\bibitem{ehrgott2005multicriteria}
M.~Ehrgott, \emph{Multicriteria optimization}.\hskip 1em plus 0.5em minus
  0.4em\relax Springer Science \& Business Media, 2005, vol. 491.

\bibitem{pillai2005perron}
S.~U. Pillai \emph{et~al.}, ``The {P}erron-{F}robenius theorem: {S}ome of its
  applications,'' \emph{IEEE Signal Processing Magazine}, vol.~22, no.~2, pp.
  62--75, 2005.

\bibitem{CGD}
\BIBentryALTinterwordspacing
CGD, ``{ESA CGD} - {P}rototype of a centralized broadband gateway for precoded
  multi-beam networks.'' [Online]. Available:
  \url{https://wwwfr.uni.lu/snt/research/sigcom/projects/esa_cgd}
\BIBentrySTDinterwordspacing

\end{thebibliography}
\end{document}